\documentclass[11pt]{article}
\usepackage{amssymb}
\usepackage{amsfonts}
\usepackage{amsmath}
\usepackage{mathrsfs}
\usepackage{amsthm}
\usepackage{latexsym}
\usepackage{hyperref}
\usepackage{amsmath}
\usepackage{amsthm}
\usepackage{amsfonts}
\usepackage{fullpage,appendix}
\usepackage{algorithm}
\usepackage{algorithmic}
\usepackage{dsfont}

\usepackage{color}
\usepackage{tikz}
\newcommand{\newref}[2][]{\hyperref[#2]{#1~\ref*{#2}}}
\newcommand{\aref}[1]{\newref[Appendix]{#1}}
\newcommand{\sref}[1]{\newref[Section]{#1}}
\newcommand{\dref}[1]{\newref[Definition]{#1}}
\newcommand{\tref}[1]{\newref[Theorem]{#1}}
\newcommand{\lref}[1]{\newref[Lemma]{#1}}

\newcommand{\cref}[1]{\newref[Corollary]{#1}}

\newcommand{\eref}[1]{\newref[Equation]{#1}}

\newcommand{\fctref}[1]{\newref[Fact]{#1}}
\renewcommand{\eqref}[1]{\hyperref[#1]{(\ref*{#1})}}

\newcommand{\eqdef}{\stackrel{\textrm{def}}{=}}
\newcommand{\etal}{et al.\ }
\newcommand{\spn}{\S^{n-1}}
\newcommand{\tm}{\tilde{m}}
\newcommand{\ignore}[1]{}
\definecolor{corlinks}{RGB}{64,128,128}
\definecolor{cormenu}{RGB}{0,37,94}
\definecolor{corurl}{RGB}{0,46,91}

\newcommand{\on}{\{-1,1\}}
\newcommand{\1}{\mathds{1}}
\newcommand{\rT}{\mathscr{T}}
\newcommand{\trho}{\tilde{\rho}}

\newcommand{\ip}[2]{\langle #1, #2 \rangle}

\newcommand{\dcdf}{\mathsf{dcdf}}

\newcommand{\E} {\mathbb{E}}

\DeclareMathOperator*{\pr}{\mathsf{Pr}} 
\newcommand{\R}{\mathbb{R}}

\newcommand{\D}{\mathcal{D}}
\renewcommand{\H}{\mathbb{H}}

\newcommand{\G}{\mathbb{G}}
\newcommand{\N}{\mathbb{N}}
\newcommand{\cN}{\mathcal{N}}

\newcommand{\C}{\mathbb{C}}

\newcommand{\SO}{\mathsf{SO}}
\newcommand{\SU}{\mathsf{SU}}
\newcommand{\sign}{\mathsf{sign}}
\renewcommand{\L}{\mathcal{L}}

\newcommand{\tnu}{\nu}

\newcommand{\poly}{\mathsf{poly}}

\renewcommand{\S}{\mathbb{S}}

\newcommand{\zo}{\{0, 1\}}

\newcommand{\eps}{\epsilon}
\newcommand{\tf}{\tilde{f}}

\newtheorem{fact}{Fact}[section]
\newtheorem{definition}{Definition}[section]

\newtheorem{theorem}{Theorem}
\newtheorem{lemma}{Lemma}[section]
\newtheorem{corollary}{Corollary}[section]

\newtheorem{claim}{Claim}[section]

\newcounter{Algo}
\floatstyle{plain}
\newfloat{myalgo}{tbhp}{mya}

{\begin{myalgo}[#1]
\refstepcounter{Algo}
\centering
\begin{minipage}{#2}
\begin{algorithm}[H]}%
{\end{algorithm}
\end{minipage}
\end{myalgo}}

\title{Almost Optimal Pseudorandom Generators for Spherical Caps}
\author{Pravesh Kothari\thanks{University of Texas at Austin, Email: kotpravesh@gmail.com} \and Raghu Meka\thanks{Email: raghuvardhan@gmail.com} }

\begin{document}

\thispagestyle{empty}
\setcounter{page}{0}
\maketitle

\begin{abstract}
 Halfspaces or linear threshold functions are widely studied in complexity theory, learning theory and algorithm design. In this work we study the natural problem of constructing pseudorandom generators (PRGs) for halfspaces over the sphere, aka spherical caps, which besides being interesting and basic geometric objects, also arise frequently in the analysis of various randomized algorithms (e.g., randomized rounding). We give an explicit PRG which fools spherical caps within error $\epsilon$ and has an almost optimal seed-length of $O(\log n + \log(1/\epsilon) \cdot \log\log(1/\epsilon))$. For an inverse-polynomially growing error $\epsilon$, our generator has a seed-length optimal up to a factor of $O( \log \log {(n)})$. The most efficient PRG previously known (due to Kane \cite{Kan12b}) requires a seed-length of $\Omega(\log^{3/2}{(n)})$ in this setting. We also obtain similar constructions to fool halfspaces with respect to the Gaussian distribution.

Our construction and analysis are significantly different from previous works on PRGs for halfspaces and build on the iterative dimension reduction ideas of \cite{KMN11, CelisRSW13}, the \emph{classical moment problem} from probability theory and explicit constructions of approximate \emph{orthogonal designs} based on the seminal work of Bourgain and Gamburd \cite{BG11} on expansion in Lie groups. 
\end{abstract}
\clearpage

\section{Introduction}
A halfspace (a.k.a.~a linear threshold function) is a function $f:\R^n \rightarrow \on$ of the form $f(x) = \sign( \sum_{i =
  1}^n w_i x_i - c) = \sign( \ip{w}{x} -c)$. Here $w_1, w_2, \ldots, w_n, c$ are reals and
$\sign(z) = 1$ for every $z \geq 0$ and $-1$ otherwise. Halfspaces are a simple class of
Boolean functions extensively studied in various contexts in computer science beginning with
threshold logic in switching circuits \cite{Der65,Hu65,LC67, Mur71,
  She69}, circuits with majority and linear
threshold gates in complexity theory \cite{GHR92, GK98,
  HMPST93,Kra96,KW91,FKLMSS01} and
voting and social choice theory \cite{Pen46, Isb69, DS79,
  TZ92}. In the recent years, they have been studied extensively in
learning theory where learning halfspaces (and functions of a few
halfspaces) is arguably the central problem and lies at the core of 
several important machine learning tools such as the perceptron \cite{Ros55},
support vector machines \cite{Vap98} and boosting \cite{FS97}.

In this paper, we study the problem of constructing explicit pseudorandom generators (PRGs) for
the class of halfspaces. Constructing PRGs for halfspaces (and more
generally, polynomial threshold functions (PTFs)) has been intensively
studied in the recent years \cite{BLY09, DGJSV10,DKN10,Kan11a,Kan11b,Kan12b,Kan12a,MZ13}. In addition to being a natural
problem in derandomization, efficient PRGs for halfspaces have
concrete applications such as derandomization of the Goemans
Williamson algorithm for max cut and 
deterministic estimation of accuracy of halfspace classifiers in
machine learning. Before proceeding, we define PRGs for halfspaces formally\footnote{Henceforth, for a multi-set $S$, $x \sim S$ denotes a uniformly random element of $S$.}:

\begin{definition}[PRG for Halfspaces w.r.t. a distribution $\nu$]
A PRG for halfspaces with respect to a distribution $\nu$ on $\R^n$, with error $\epsilon$, is a function
$G:\zo^r \rightarrow \R^n$ such that for every halfspace $h:\R^n
\rightarrow \on$,
 $$ \left| \E_{x \sim \nu}[ h(x)] - \E_{y \sim \zo^r}[
h(G(y))] \right| \leq \epsilon .$$
The parameter $r$ is called the \emph{seed-length} of the PRG. $G$ is said
to be explicit, if $G(y)$ can be computed in time polynomial in $n$.
\end{definition}
Alternatively, $G$ is said to ``fool halfspaces with error $\epsilon$ w.r.t. $\nu$'' when
it satisfies the definition above. We will skip the explicit reference to $\nu$ when the distribution is clear from context. In this work, we focus on the case when $\nu$ is one of two natural distributions: the
uniform distribution on the sphere $\S^{n-1} \subseteq \R^n$ and the
$n$ dimensional spherical Gaussian distribution $\cN(0,1)^n$. Halfspaces are also referred to as \emph{spherical caps} when seen as functions on $\S^{n-1}$. A simple probabilistic argument shows that there \emph{exist}
PRGs with seed-length $r= 2\log{(n)} + 2\log{(1/\epsilon)} + O(1)$ for both these distributions.

The question of constructing explicit PRGs for spherical caps was first studied in computational geometry in the form of \emph{discrepancy minimization} for spherical caps (cf.~\cite{Chazellebook}). This line of inquiry led to the seminal work of Lubotzky, Philips and Sarnak (\cite{LPS87}) who used \emph{Ramanujan expanders} to construct a PRG for spherical caps over $\S^2$ (in our language) with a seed-length of $3\log(1/\epsilon) + O(1)$. More recently, Diakonikolas \etal \cite{DGJSV10} showed that bounded independence fools halfspaces w.r.t. the uniform distribution on the Boolean hypercube, giving a PRG with seed-length $O(\log{(n)} /\epsilon^2)$. Subsequently, Karnin, Rabani and Shpilka \cite{KarninRS12} developed a PRG for spherical caps with a seed-length of $O(\log n + \log^2(1/\epsilon))$. The best current result due to Kane \cite{Kan12b} gives a PRG for spherical caps with a seed-length of $O( \log{(n)} + \log^{3/2}{(1/\epsilon)})$. 

Despite the long line of works, the seed-length for the best PRGs for halfspaces remains off by poly-logarithmic factors in $n$ for low error regimes (i.e. $\epsilon \approx 1/\poly(n)$). In this work, we resolve this question for spherical caps and give a construction with seed-length optimal up to a factor of $O( \log \log {(n)})$.

%

\begin{theorem}[Main Theorem]
There exists a PRG for halfspaces with seed-length
$O(\log{(n)} + \log{(1/\epsilon)} \cdot \log \log {(1/\epsilon)})$ for error $\epsilon$, on the uniform
distribution on sphere $\S^{n-1}$. 
\end{theorem}
Our result extends to fool halfspaces with respect to Gaussian distributions as well. 

As we describe next, our construction departs significantly from the previous work on constructing PRGs and introduces new ingredients which may be useful elsewhere. In particular, our construction uses an iterative dimension reduction approach as in the works of \cite{KMN11, CelisRSW13} and makes use of explicit constructions of approximate \emph{orthogonal designs} which are related to quantum analogues of classical $k$-wise independence and expanders \cite{AmbainisE07, Harrow08, HL09, HH09, GrossEisert08, BHH12}. The analysis of the construction is motivated by the classical \emph{truncated moment problem} from probability theory. 

\subsection{Outline of constructions}
We now give a high-level description of our construction of the PRG and its analysis for spherical caps. For simplicity, we focus on the case of $\epsilon = 1/\poly(n)$ aiming to achieve a seed-length of $O( \log{(n)} \log \log {(n)})$. Consider a vector $w \in \R^n$ describing a half space $\sign( \ip{w}{x} -c)$. Without loss of generality, we can assume that the vector $w$ is normalized so that $||w|| = \sqrt{ \sum_{i = 1}^n w_i^2} = 1$. It follows immediately from the definitions that to construct a PRG for spherical caps, it suffices to construct a generator $G:\zo^r \rightarrow \R^n$ such that 
$$ \dcdf(\ip{w}{G(y)}, \ip{w}{x}) < \epsilon,$$
where $y \sim \zo^r$ and $x \sim \spn$ and $\dcdf$ denotes the CDF or
Kolmogorov distance between real-valued random variables. 

Let $X$ be the random variable $\ip{w}{x}$ for $x \sim \spn$. We can think of $X$ as obtained by projecting $w$ on to a uniformly random direction. or equivalently, a random one-dimensional \emph{subspace} of $\R^n$. Our construction will exploit this geometric viewpoint by using the following trivial observation: 

\begin{fact}\label{fct:iterated}
For any $1 \leq m \leq n$, first picking a random $m$-dimensional subspace $V \subseteq \R^n$ and then picking a random one-dimensional subspace of $V$ gives the same distribution as picking a random one-dimensional subspace in $\R^n$.  
\end{fact}

We use the above observation by iteratively projecting the vector $w$ into $\sqrt{n}$ dimensions, and then to $n^{1/4}$ dimensions and so forth until we work in a space of dimension $\Theta(\log n)$. Once we are down to vectors of dimension $\Theta(\log n)$, we use a direct approach to project down to a one-dimensional subspace. We will ensure that each one of these projections can be carried out with $O(\log(n/\eps))$ random bits and preserves the properties (including closeness in CDF distance) that we want. Thus, the total randomness used by our generator will be $O(\log{(n/\epsilon)}\cdot \log \log {(n)} )$ random bits.

To make the above outline concrete let us introduce a central definition\footnote{The use of $\sqrt{n}$ below is somewhat arbitrary and any $n^c$ for $c < 1$ would suffice for us. We choose $\sqrt{n}$ to reduce the number of parameters.}:

\begin{definition}[Pseudorandom projections (PRP)]\label{def:prp}
 Let $Q$ be a uniformly random projection\footnote{See Definition \ref{def:uniform-projection} for a more precise formulation} from $\R^{n} \rightarrow \R^{\sqrt{n}}$ and $x \sim \S^{\sqrt{n}}$. A distribution $\D$ on projections from $\R^{n}
\rightarrow \R^{\sqrt{n}}$ is an $\epsilon$-pseudorandom projection (PRP) if the following holds for all $w \in \R^n$ and $P \sim \D$: 
$$\dcdf(\ip{Pw}{x}, \ip{Qw}{x}) \leq \epsilon.$$

The number of random bits required to sample a $P$ distributed as $D$ is the \emph{seed-length} of the PRP.
\end{definition}

Roughly speaking, the above definition says that projecting any vector $w$ to $\sqrt{n}$ dimensions using our PRPs and then projecting to a truly random one-dimensional subspace is indistinguishable from using truly random projections. 

Before describing our construction of PRPs let us note how they can be used for constructing PRGs fo spherical caps. As described above, we use our PRPs $O(\log \log n)$ times to project our vector down to $\Theta(\log^2{(n)})$ dimensions. At this point, we invoke the PRG of Impagliazzo \etal \cite{INW94} for space bounded machines (that has a seed-length of $O(\log{(d)} \cdot \log {(1/\epsilon)})$ for fooling halfspaces in $d$ dimensions with error $\epsilon$). To bound the error we just use a union bound to bound the errors of all projection steps and use \fctref{fct:iterated}. 

In order to extend the construction above for halfspaces on the Gaussian distribution, we observe that $\ip{w}{g}$ for $g \sim \cN(0,1)^n$
is identically distributed as $||g||_2 \cdot \ip{w}{x}$ for $x \sim
\S^{n-1}$ and $||g||_2$, the length of the random (and independent of $x$) Gaussian vector $g$
(equivalently, a $\chi$-distributed random variable with $n$ degrees
of freedom). If $G$ is the PRG described above for spherical caps,
we obtain that $ \chi_{n,\epsilon}\cdot G$ is a PRG for halfspaces on the Gaussian distribution, where $\chi_{n,\epsilon}$ is obtained by discretization of a $\chi$-distributed random variable with $n$ degrees of freedom.

We next describe our construction of explicit PRPs.
  
\subsection{Pseudorandom projections and the classical moment problem}
Note that in our definition of the PRPs, the vector $x$ is truly random over $\S^{\sqrt{n}}$; our analysis will exploit this. At a high level, this helps us because even if our PRPs do not approximate truly random projections very well, the {\sl extra} randomness from $x$ will be enough to mask these imperfections sufficiently.

To describe the intuition more concretely, let us setup some notation. Fix a test vector $w \in \R^n$ and $\D$ be our candidate PRP as in \dref{def:prp} and $P \sim \D$. Let $Q$ be a truly random projection from $\R^n$ to $\R^{\sqrt{n}}$. Let $X = \ip{Qw}{x}$ and $Y = \ip{Pw}{x}$; our goal is to design $\D$ so that $X,Y$ are close in CDF distance. Our plan to bound $\dcdf(X,Y)$ is to match the low order moments of $X$ and $Y$. This idea relates to the {\sl classical moment problem}\cite{Ak65}: when do the moments of a (univariate) random variable over a specified range uniquely identify the random variable? In our context, the more relevant question is the {\sl truncated moment problem}: given two real-valued random variables with matching first $k$ moments, how close (under various metrics) are the two random variables?

There is a rich history behind these two questions (see for example, \cite{Ak65}, \cite{Las12}). Unfortunately, the results from the probability literature are quantitatively too weak for us: in most of these general results one needs to match $(1/\epsilon)^{\Omega(1)}$ moments (see \cite{KlebanovM1980} and the discussion after Lemma \ref{lem:moment2cdf}) to get error $\epsilon$ which we cannot afford as we aim for $\epsilon$ which is polynomially small. 

Our main idea is to exploit the additional structure of the random variables $X,Y$. For example, the random variable $Y$ above has nicely behaved moments (i.e., not growing too fast) and also has a smooth, well-behaved (read: bounded derivatives) probability  density function (PDF). To exploit this, let us be more concrete. Note that the distribution of random variables $X,Y$ only depend on the norms $\|Pw\|_2, \|Qw\|_2$ respectively (because of the rotational symmetry of $x \sim \S^{\sqrt{n}}$). Thus, if we let $X' = \|Qw\|_2$ and $Y' = \|Pw\|_2$, we can write $X = X' \cdot Z$ and $Y = Y' \cdot Z$, where $Z$ is the random variable obtained by projecting $x \sim \S^{n-1}$ to a fixed direction (say $e_1$).

It is not hard to see that $Z$ has a smooth pdf and that $X'$ has
well-behaved moments (which can be controlled by
hypercontractivity). We show that whenever the random variables $X',Z$
satisfy these reasonable conditions, if, in addition, the first $k$ (even order)
moments of $Y'$ match the corresponding first $k$ moments of $X'$, then, $X,Y$ are close within an error that is exponentially small in $k$ (the base of the exponent depending on the moments of $X'$ and smoothness of $Z$). This result fits into the
general principle where matching moments with some additional
structure can be used to get much stronger quantitative guarantees on closeness of
distributions; for example, \cite{DP13,AHK12} show similar stronger quantitative bounds for various mixture models.

The above arguments reduce the problem of designing PRPs to that of constructing a distribution $\D$ over projections from $\R^n$ to $\R^{\sqrt{n}}$ such that for $P \sim \D$, and any $w \in \R^n$, the moments of $\|Pw\|_2$ are {\sl approximately} what they should be for a corresponding truly random projection. As it turns out, such distributions have been studied in quantum computing \cite{AmbainisE07,Harrow08,GrossEisert08,HH09,HL09,BHH12, ABW09} under the label {\sl orthogonal designs}\footnote{One usually looks at {\sl unitary designs} in quantum computing, but we ignore this distinction for this high-level discussion.}. We discuss them next.
\subsection{Orthogonal designs}
Orthogonal designs can be seen as generalizations of standard tools in pseudorandomness like $k$-wise independence and almost $k$-wise independence to the ``uniform" (Haar)  distribution over rotation matrices. Let $\SO(n)$ denote all orthogonal matrices in $\R^{n \times n}$ and let $\H$ denote the Haar measure on $\SO(n)$. By polynomials on $\SO(n)$, we mean functions that are polynomials in the entries of matrices from $\SO(n)$.

\begin{definition}[approximate orthogonal $t$-design]
A distribution $\D$ on a finite subset of matrices from $SO(n)$ is said
to be an $\epsilon$-approximate $t$-design (in $n$ dimensions) if for every polynomial
$p:\SO(n) \rightarrow \R$ of degree at most $t$ such that $||p||_1 = 1$ (where $||p||_1$ denotes the sum of absolute values of coefficients of $p$), $$| \E_{\D} [ p ] - \E_{\H}[p]| \leq \frac{\epsilon}{n^t}.$$ We say that $\D$ is an \emph{explicit} orthogonal design if there is a $\poly(n)$ time procedure to sample a matrix according to $\D$. The number of bits of randomness used to sample a matrix according to $\D$ is called its \emph{seed-length}.
\end{definition} 
It is not too hard to show using the definitions and the arguments outlined from the previous section, that taking the matrix of first $\sqrt{n}$ rows of an approximate orthogonal $t$-design one gets a PRP with the same seed-length and error which is $\epsilon$. If we think of fixing the error $\epsilon$ (the dimension changes for us as we recurse), to get PRPs with an error of $\epsilon$ we need a $t$-design for $t \approx O\left( \log{(1/\epsilon)}/\log {(n)} \right)$.

In particular, to get PRPs with nearly-optimal seed-length, it suffices to get approximate $t$-designs with the near-optimal seed-length. The existence of  \emph{finite} orthogonal (or unitary) designs follows from the general results of Seymour and Zaslavsky \cite{SZ84}. Harrow and Low \cite{HL09} observe that one can modify the argument of Ambainis et. al. \cite{ABW09} to show that there exist $t$-designs with optimal (up to constants) parameters. Our application, however, requires efficient \emph{explicit} constructions of these objects and we use the work of Brandao \etal \cite{BHH12} who showed that a recent breakthrough result of Bourgain and Gamburd \cite{BG11} on expansion in Lie groups implies a construction of approximate orthogonal designs for $t \leq \Theta(n)$. This gives us orthogonal designs with optimal seed-length (up to constant factors).\footnote{As some of the parameters important in our setting are not specified in \cite{BHH12} and we work over real matrices as opposed to Hermitian ones in \cite{BHH12}, we give an analysis of the construction of $t$-designs from the expansion results of Bourgain and Gamburd \cite{BG11} in Section \ref{sec:orthdesign}.}
\subsection{Other related work}
There is a vast body of work in probability on the \emph{generalized moment problems}, beginning with Stieltjes \cite{Stj1894} with a first systematic study appearing in the work of Akhiezer \cite{Ak65} under the name of \emph{classical moment problem}. The ideas are extremely useful in applications in a number of different areas (see the recent textbook of Lasserre \cite{Las12} for a host of applications). For a survey of various approaches to the moment problem, see the text by Landau \cite{Lan87}.

The question of distance (in various metrics) between probability distributions that have (approximately) matching low-degree moments is also well studied, see, for example \cite{KKR88}, where the principle measure of distance used is the $\lambda$-\emph{metric}. It is possible (see for example \cite{KKM13}) to convert these bounds into the more standard CDF (or Levy) distance bounds using known results \cite{Gab76}.

The idea of using stepwise projections in order to reduce the amount of randomness required in each step was successfully employed in constructing almost optimal (with respect to randomness) explicit Johnson-Lindenstrauss (JL) embeddings by Kane et. al. \cite{KMN11}. The analysis in \cite{KMN11} is also based on matching the low-order moments of the lengths of the projections in each step. Their argument, though, is different and simpler as a JL family needs
to satisfy only a tail bound condition and one can move from matching low-order moments to tail bounds under simple conditions on the random variables. In contrast, the connection between matching low order moments and CDF distance, as explained above, doesn't hold in general and we crucially exploit the additional smoothening effect of mixing with an independent well behaved random variable to obtain the low errors we need.

Very recently, Gopalan, Kane and Meka \cite{GopalanKM14} gave a PRG for halfspaces whose coefficients are in $\{1,0,-1\}$ w.r.t the Boolean hypercube with a seed-length of $O((\log (n/\epsilon)) \cdot \poly\log(\log(n/\epsilon)))$; this is incomparable to ours and their methods do not seem to apply in our setting. Their work also uses the iterative dimension reduction approach as in \cite{KMN11} but the actual construction and its analysis are very  different from ours.




\section{Preliminaries}\bibliographystyle{plain}
We start with some notations: 
\begin{enumerate}
\item For any vector $V = (V_1, V_2, \ldots, V_n) \in \R^n$, $\|V\| = \|V\|_2 =\sqrt{ \sum_{i = 1}^n V_i^2}$ denotes its Euclidean norm. We will use the same notation for the norm (induced by the inner product) of elements of any infinite dimensional Hilbert space.
\item For any matrix of reals $M$, $M^{\dagger}$ denotes its transpose, $\|M\|$ its spectral norm (largest singular value) and $\|M\|_2 = \sum_{i,j} M_{i,j}^2$, its Frobenius (or $2$) norm. 
\item $\cN(0,1)^n$ denotes the spherical multivariate Gaussian distribution in $n$-dimensions or alternatively, the product Gaussian measure on $\R^n$ with PDF at $X \in \R^n$ given by $(\frac{1}{2\pi})^{-n/2} e^{-\frac{1}{2} \|X\|_2^2}$. 
\item $\chi_n$ ($\chi$ random variable with $n$-degrees of freedom) denotes the positive real-valued random variable distributed as $Y = \|X\|_2$ where $X \sim \cN(0,1)^n$.  
\item The Gamma function $\Gamma:\R \rightarrow \R$ is defined at any $t \in \R$ by the integral $\Gamma(t) = \int_{0}^{\infty}
x^{t-1} e^{-x} dx$ and equals $(t-1)!$ whenever $t$ is a positive integer.
\end{enumerate}

\begin{definition}[CDF or Kolmogorov Distance]
Let $X$ and $Y$ be random variables on some domain $D$ with cumulative distribution functions (CDFs) $P_1$ and $P_2$ respectively. The CDF distance between $X$ and $Y$ is defined as $\dcdf(X,Y) = \sup_{z \in D} |P_1(z) - P_2(z)|.$
\end{definition}
\subsection{Random rotation and projection matrices}
$\SO(n)$ denotes the group (under matrix multiplication) of all real orthogonal $n \times n$ matrices. There is a unique probability measure on $\SO(n)$ invariant under matrix multiplication (on the left or right) by matrices in $\SO(n)$ and is called as the \emph{Haar} distribution (see \sref{app:sec-Haar} for a brief overview). 

$\S^{n-1} \subseteq \R$ denotes the sphere of radius $1$ in $n$ dimensions. The uniform distribution on $\S^{n-1}$ is the unique probability distribution on $\S^{n-1}$ invariant under the action of matrices from $\SO(n)$. One can think of $\S^{n-1}$ as the set of all one-dimensional subspaces of $\R^n$ corresponding to each direction (unit vector) it contains. More generally, let $G_{n,t}$ (the \emph{Grassmanian}) denote the set of all $t$ dimensional subspaces of $\R^n$. There exists a unique probability (Haar) measure $\H$ on $G_{n,t}$ such that given any subspace $W \in G_{n,t}$, $\H(W) = \H( O \cdot W)$ where $O \in \SO(n)$ is any rotation matrix. By a uniformly random subspace of $t$ dimensions, we mean an element of $G_{n,t}$ drawn according to the distribution with the PDF $\H$.

To project any given vector $w \in \R^n$ on to a random subspace from $G_{n,t}$, we can draw a matrix $Q$ from $\SO(n)$ distributed according to the Haar measure on $\SO(n)$ and then select the sub matrix formed by the first $t$ rows of $Q$ to obtain $R$. Then, $Rw$ yields the required random projection. It is a well known fact that the distribution so generated is identical to the Haar measure on $G_{n,t}$.
\begin{definition}[Uniformly Random Projection Matrix] \label{def:uniform-projection}
Let $O$ be drawn from the Haar distribution on $SO(n)$. For any $m \leq n$ the uniformly random projection matrix from $\R^n$ to $\R^m$ is defined by $Q_{m,n} \in \R^{m \times n}$, the matrix obtained by taking the first $m$ rows of $O$.  
\end{definition}

A standard and useful property of uniforly random projections is that one can perform them stepwise: for any $w \in \R^n$ and $\tm \leq m \leq n$ let $Q_{m,n}$ and $Q_{\tm,m}$ be independent random projections from $R^n \to \R^m$ and $\R^m \to \R^{\tm}$. Let $Q_{\tm,n}$ be a uniformly random projection from $\R^n \to \R^{\tm}$. Then, $Q_{n,\tm} \cdot w \in \R^{\tm}$ and $Q_{\tm,m} \cdot Q_{m,n} \cdot w \in \R^{\tm}$ are identically distributed. We refer the reader to the lecture notes by Vershynin \cite{Ver11} for background on random projections and to the text \cite{Bu11} for background on the orthogonal group.


\section{PRGs for spherical caps from pseudorandom projections }
In this section, we describe our main result giving a PRG for fooling
spherical caps over $\S^{n-1}$ with nearly optimal seed-length:
\begin{theorem}[PRG for Spherical Caps] \label{sphericalPRG}
There exists a PRG for spherical
caps on $\R^n$ with error at most $\epsilon$ and a seed-length of $O(\log {(n)} + \log \log {(1/\epsilon)} \log(1/\epsilon))$.
\end{theorem}
We will prove the above result assuming we have constructions of appropriate PRPs as defined in \dref{def:prp}; we show how to construct PRPs in the subsequent sections. 

\begin{theorem}[See \sref{sec:PRPs}]
Fix any $\epsilon > 0$. Then, for any $m = \Omega( \log^2{(1/\eps)})$, there exists an $\epsilon$-PRP from $\R^m
\rightarrow \R^{\sqrt{m}}$ with a seed-length of $O( \log{
  (m/\epsilon)})$. \label{thm:sphericalprojection}
\end{theorem} 

When working on small dimensions ($m \sim \log{(1/\epsilon)}$), we will use the PRG for halfspaces based on the construction for small-space machines due to Nisan \cite{Nis92} and Impagliazzo et. al. \cite{INW94} as observed in \cite{MZ10}.

\begin{fact}
There exists a PRG $G_{INW}:\zo^s \to \S^{m-1}$ that fools spherical caps in $m$ dimensions within
error $\epsilon$ using a seed-length of $s = O(\log{(m)} \cdot
\log{(1/\epsilon)})$. \label{fact:INWPRG}
\end{fact}

We are now ready to prove \tref{sphericalPRG}. As described in the introduction, the basic idea is to use the PRPs to iteratively reduce the dimension of the space and when the dimension is small enough, we can apply \fctref{fact:INWPRG}. 

We now describe our generator construction. Fix $\epsilon' = \Theta(\epsilon/\log \log{(1/\epsilon)}) > 0$ and for $i = 0,1,2,\ldots,t-1$ let $n_i = n^{1/2^i}$. Set $t = O(\log \log { (1/\epsilon')})$ such that $n_t = \max \{ n^{\Theta((1/\log{(1/\epsilon')}))}, \Theta(\log^2{(1/\epsilon')}) \}$. Let $\D_i$ denote a PRP from $\R^{n_i} \to \R^{n_{i+1}}$ with error $\epsilon'$ and let $G_{INW}:\zo^s \to \S^{n_t-1}$ be the generator given by \fctref{fact:INWPRG} with error $\epsilon'$ and $s = O( \log{(n_t)} \cdot \log{(1/\epsilon')})$. 
Our generator for fooling spherical caps is defined as follows:
\begin{itemize}
\item Sample $P_i \sim \D_i$ for each $i < t$ and let $X_t = G_{INW}(y)$ for $y \sim \zo^s$, all random draws being independent of each other.
\item Output
\begin{equation}
  \label{eq:maingen}
X_0 =  (P_0^\dagger P_1^\dagger \cdots P_{t-1}^\dagger X_t) \in \S^{n-1},
\end{equation}
\end{itemize}
Note that because $P_i$ is a projection, the rows of $P_i$ are orthonormal and as $\|X_t\| = 1$ it follows that $\|X_0\| = 1$ so that $X_0 \in \S^{n-1}$. The total seed-length of the generator is the number of bits needed for each of the PRPs and $G_{INW}$:
\begin{align} 
 \sum_{i=1}^{t-1} O(\log(n_i/\epsilon')) + O(\log(n_t) \cdot \log(1/\epsilon')) = O(\log{(n)} + \log \log{(1/\epsilon)} \cdot \log{(1/\epsilon)}).  \label{eq:maingenseed}
\end{align}
We will show that the above generator fools spherical caps over $\S^{n-1}$ with error $O(\log \log{(1/\epsilon)}\cdot \epsilon') \leq \epsilon$.  
\begin{proof}[Proof of \tref{sphericalPRG}]
We will show that the generator defined by \eref{eq:maingen} fools spherical caps w.r.t. the uniform distribution on $\S^{n-1}$. The proof is a simple inductive application of the definition of PRPs. We first set up some notation.
For $0 \leq i \leq t$, let $Y_i \sim \S^{n_i-1}$ and $X_i = (P_i^\dagger P_{i+1}^\dagger \cdots P_{t-1}^\dagger X_t)$.  Let $Q_i \in \R^{n_{i+1} \times n_i}$ be a truly random projection from $\R^{n_i}$ to a uniformly random $n_{i+1}$-dimensional subspace inside it. Fooling spherical caps is equivalent to showing that $\dcdf(\ip{w}{X_0},\ip{w}{Y_0})$ is small. We will show this by induction on $0 \leq i \leq t$. For brevity, for real-valued random variables $Z,Z'$, we write $Z \approx_\delta Z'$ if $\dcdf(Z,Z') \leq \delta$.
\begin{claim}
  For $0 \leq i \leq t$, and all $v \in \R^{n_i}$, $$\dcdf(\ip{v}{X_i}, \ip{v}{Y_i}) \leq (t-i + 1) \epsilon'.$$
\end{claim}
\begin{proof}
For $i = t$, the statement in the claim follows from the definition of $X_t$ and \fctref{fact:INWPRG}. Suppose the claim is true for $i = j+1$ for some $0 \leq j \leq t-1$. Observe that $X_j = P_j^\dagger X_{j+1}$. Fix $v \in \R^{n_j}$. Now,
$$\ip{v}{X_j} = \ip{v}{P_j^\dagger X_{j+1}} = \ip{P_j v}{X_{j+1}} \approx_{(t-j)\epsilon'} \ip{P_j v}{Y_{j+1}} \approx_\epsilon' 
\ip{Q_j v}{Y_{j+1}} = \ip{v}{Q_j^\dagger Y_{j+1}} = \ip{v}{Y_j},$$
where the first $\approx$ follows from the inductive hypothesis (and the fact that $P_j$ is independent of $X_{j+1}$) and the second $\approx$ follows from the definition of PRP. Therefore, $\dcdf(\ip{v}{X_j}, \ip{v}{Y_j}) \leq (t-j+1)\epsilon'$. The claim now follows by induction.
\end{proof}
Therefore, for any $w \in \R^n$, $\dcdf(\ip{w}{X_0}, \ip{w}{Y_0}) \leq (t+1) \epsilon' \leq \epsilon$. Using the bound on the seed-length from \eref{eq:maingenseed}, the theorem follows.
\end{proof}

\subsection{PRGs for halfspaces with respect to the Gaussian distribution} \label{app:gaussian-proj}
Next, we describe and analyze the PRG for Halfspaces w.r.t. the spherical Gaussian distribution on $\R^n$. The generator will just be an appropriate scalar multiple of the generator for fooling spherical caps from \tref{sphericalPRG}.

\begin{theorem}[PRGs for Halfspaces of Gaussians] \label{GaussianPRG}
Fix $\epsilon > 0$. There exists a PRG for halfspaces
w.r.t. the spherical Gaussian distribution with error at most $\epsilon$ and seed-length $s = O( \log n + \log \log {(1/\epsilon)} \cdot \log(1/\epsilon))$. 
\end{theorem}

To prove the theorem we shall use the following simple fact.
\begin{lemma}\label{eps-approx-chi}
For every $n \geq 1$, and $\delta > 0$, there exists a random variable $\chi_{n,\delta}$ samplable efficiently with $O(\log n + \log(1/\delta))$ bits that approximates the random variable $\chi_n$: $$\dcdf(\chi_{n,\delta}, \chi_n) \leq \delta.$$ 
\end{lemma}

\begin{lemma}\label{lm:productrv}
  Let $U, V_1, V_2$ be independent random variables. Then, $\dcdf(U \cdot V_1, U \cdot V_2) \leq \dcdf(V_1, V_2)$.
\end{lemma}

\begin{proof}[Proof of \tref{GaussianPRG}]
The theorem follows from the following black box reduction and \tref{sphericalPRG}.

Fix $\epsilon' = \epsilon/2$ and let $G:\zo^s \to \S^{n-1}$ be the generator from previous subsection that $\epsilon'$-fools spherical caps w.r.t. the uniform distribution on $\S^{n-1}$. Let $\chi_{n,\epsilon'}$ denote a $\epsilon'$-approximating random variable for $\chi_n$ from \lref{eps-approx-chi}. 

Our generator samples $\chi_{n,\epsilon'}$ and $x \sim \zo^s$ independently and outputs $\chi_{n,\epsilon} \cdot G(x)$. Let $Y \sim \S^{n-1}$. It is a standard fact that $\cN(0,1)^n \equiv \chi_n \cdot Y$ (in law). Now, for any $v \in \R^{n_t}$, 
\begin{multline*}
\dcdf(\langle \chi_{n,\epsilon'} \cdot G(x), v \rangle, \langle \chi_n \cdot Y, v \rangle) \leq \dcdf(\langle \chi_{n,\epsilon'}\cdot G(x), v \rangle, \langle \chi_n\cdot G(x), v \rangle) + \dcdf(\langle \chi_n \cdot G(x), v \rangle, \langle \chi_n \cdot Y, v \rangle) \leq \\
\dcdf(\chi_{n,\epsilon'}, \chi_n) + \dcdf(\langle G(x), v\rangle, \langle Y, v\rangle) \leq 2\epsilon' = \epsilon,
\end{multline*}
using \lref{lm:productrv}. The theorem now follows.
\end{proof}

 The objective of the following two sections is to prove \tref{thm:sphericalprojection}.
\section{From matching moments to CDF distance}
\label{sec:moments}
In this section, we give quantitative bounds for the truncated moment problem in terms of the CDF distance: given random variables that have approximately equal low order moments, we derive strong upper bounds on the CDF distance between them. Our bounds are stronger (and crucial to obtaining near optimal seed-lengths for our generators) than those obtained from the general results in probability (see for e.g. \cite{KKR88}, \cite{Lan87}) but require stronger analytic properties of the random variables. The results from this section will be used to analyze our construction of PRPs in the next section.

Let $Z$ be a real-valued random variable with mean $0$ and CDF $F:\R
\to [0,1]$. 
Our goal here is to show
that if $X,Y$ are two non-negative random variables with approximately
equal first $k$ moments, then the mixtures $X \cdot Z$ and $Y \cdot Z$ are close in CDF distance. We will show that if $F$ is sufficiently \emph{smooth} and $Z$ satisfies reasonable tail bounds, then the cdf distance between $X \cdot Z$ and $Y \cdot Z$ is exponentially small in $k$.  We first formally define (approximate) moment matching
random variables. 
\begin{definition}[Moment-Matching Random Variables]
Random variables $X$ and $Y$ are said to be
$(k,\epsilon)$-approximate moment matching if for every polynomial $p$
of degree at most $k$,
\begin{equation} 
\label{approx-moment-matching} 
|\E[p(X)]-\E[p(Y)] | \leq \epsilon \cdot \|p\|_1, \end{equation}
where $\|p\|_1$ is the sum of absolute values of the coefficients of $p$. 
\end{definition}

We are now ready to describe the main technical result of this section:
\begin{lemma} \label{lem:moment2cdf}
 Let $Z$ be a mean $0$, symmetric random variable with an infinitely differentiable CDF $F:\R \to [0,1]$. Let $X,Y$
 be two non-negative $(2k,\epsilon)$-approximate moment matching random
 variables for even $k$. Let $\mu^2 = \E[X]  \geq 1$ and $\mu_i = \E[\left|X
 - \mu^2\right|^i]$ for every $i > 1$. Then, for all $\delta \leq 1/2$, 
 $$\dcdf(\sqrt{X} \cdot Z, \sqrt{Y} \cdot Z) \leq \delta + 2^{\Theta(k)} \Delta_k \cdot ( \frac{\mu_k}{\mu^{2k}} + \frac{\sqrt{\mu_{2k}}}{\mu^{2k}} + \epsilon),$$
where  $\Delta_k = \max\left(1, k^{\Theta(k)} \left( F^{-1}(1-\delta) \right)^k \left(\max_{\ell \leq
      k} \sup_{ 2\sqrt{2/3} \leq \frac{z}{F^{-1}(1-\delta)} \leq 2\sqrt{2} } |F^{(\ell)}(z)|\} \right)\right)$.
%
%
\end{lemma}

It is instructive to compare the error bounds above with the following estimate (due to Klebanov and Mkrtchyan) of CDF distance between random variables that have identical low order moments (the statement for approximately equal low order moment is similar but more cumbersome).
\begin{fact}[Klebanov and Mkrtchyan \cite{KlebanovM1980}, Theorem 1 and Remark 2]
Let $F,f$ and $G,g$ be the CDFs and PDFs of real-valued random variables $X$ and $Y$. Suppose $f$ is bounded and that $X,Y$ have identical finite first $2m$ moments given by $\mu_1, \mu_2, \ldots, \mu_{2k} < \infty$ such that $\mu_2 = 1$. Let $\beta_k = \sum_{i = 1}^{k} \mu_{2i}^{1/2i}$. Then, for a universal constant $C$, $$\dcdf(X,Y) \leq C\sup_{t \in \R} f(t) \beta^{-1/4}_{k-1}.$$ \label{fact:KlebanovM}
\end{fact}

We can get a sense of how large $k$ must be to obtain an error smaller than a given $\epsilon$ using \fctref{fact:KlebanovM}: consider the example of $X = \|Qw\|^2$ where $Q$ is a uniformly random projection matrix from $\R^{n} \to \R^{\sqrt{n}}$. In this case, after scaling $X$ to make $\mu_2 = 1$, we have that $\mu_{t}^{1/t} = t^{\Omega(1)}$ (follows from an argument similar to the proof of \lref{lem:moment-matching-verify}). Thus, to obtain an error of $\epsilon$, \fctref{fact:KlebanovM} would require matching moments of $X$ and $Y$ of order up to $(1/\epsilon)^{\Omega(1)}$. In contrast,  in \sref{sec:PRPs}, in the proof of \tref{thm:sphericalprojection} using \lref{lem:moment2cdf}, we will show that for  $n = \Omega(\log{(1/\epsilon)})$, approximately matching only a constant number of moments of $X$ and $Y$ will be enough to give us a CDF distance bound of at most $\epsilon$.

We now move on to the proof of \lref{lem:moment2cdf}.
We first collect a few simple facts from elementary analysis that will be useful in our proof of \lref{lem:moment2cdf}. We give proofs for these results in \sref{app:defproofs} of the Appendix. First we note a bound on the magnitude of derivatives of compositions of infinitely differentiable functions:
\begin{fact}\label{fct:derivative} 
 Let $g,h:\R \to \R$ be infinitely differentiable functions. Then,
$$|(g \circ h)^{(k)}(x)| \leq (k!)^2 \cdot \left(\max_{\ell \leq k} |g^{(\ell)}(h(x))|\right) \cdot \left(\max_{\ell \leq k} |h^{(\ell)}(x)|\right)^k.$$
\end{fact}

We will also need the following bound on the derivatives of $1/\sqrt{1+x}$ when $x > -1/2$:
\begin{fact}\label{fct:sqroot}
For $h:(-1,\infty) \to \R$ be defined by $h(x) = 1/\sqrt{1+x}$. Then, $|h^{(k)}(x)| \leq k!$ for $-1/2 < x$. 
\end{fact}

Finally, we write the CDF of a product of two random random variables as a convenient expression:
\begin{fact} \label{lem:basic}
Let $F:\R \rightarrow [0,1]$ be the CDF of a random variable $Z$. For a positive random variable $V$
independent of $Z$, the CDF $G$ of $V \cdot Z$ at any $t \in \R$ is given
by: $$G(t) =
\E_{V}[ F(t/V)].$$
\end{fact}

We are now ready to prove \lref{lem:moment2cdf}.

\begin{proof}[Proof of \lref{lem:moment2cdf}]
Fix $t \in \R$. Using \fctref{lem:basic}, the lemma amounts to obtaining an upper bound on   $$\sup_{t} \left| \E_X[F(t/\sqrt{X})] - \E_Y[F(t/\sqrt{Y})]\right|.$$ Ideally, we'd like to show that for every $t$, $F(t/\sqrt{A})$ is well approximated by a low-degree polynomial (in $A$). We can then invoke the approximate moment matching property of the pair $X,Y$ to bound the difference in the expectations of $F(t/\sqrt{X})$ and $F(t/\sqrt{Y})$. This, however turns out to be too strong. Instead, we will show that $F$ is well approximated by low-degree polynomials \emph{whenever} $X,Y$ do not deviate too far from their expectations. Towards this goal, we first set some notation:
\begin{enumerate}
\item $\tnu^2 \eqdef \E[Y]$. 
\item $\hat{X} \eqdef \frac{X-\mu^2}{\mu^2}$, $\hat{Y} \eqdef \frac{Y -
  \mu^2}{\mu^2}$. 
\item $g = g_t:(-1,\infty) \to \R$ by $g(z) \eqdef
F(\frac{t}{\mu(\sqrt{1+z})})$. 
\end{enumerate}
Now, $\E[\hat{X}^k] \leq \frac{\mu_k}{\mu^{2k}}$ and thus, the statement of the theorem is trivial
if $ \frac{\mu_k}{\mu^{2k}} \geq 1$. Thus, suppose otherwise, in which case $\E[\hat{X}^{\ell}] \leq 1$ for any $\ell \leq 2k$.

Now, $F(\frac{t}{\sqrt{X}}) = F(\frac{t}{\mu \cdot \sqrt{(1+\hat{X})}}) =  g(\hat{X})$ and similarly, $F(\frac{t}{\sqrt{Y}}) = F(\frac{t}{\mu \cdot \sqrt{(1+\hat{Y})}}) = g(\hat{Y})$. Next, we describe the approximating polynomial for $g$ at $0$, which will be obtained by truncating the Taylor expansion of $g$. To bound the error of approximation, we will need to bound the derivatives of $g$. This can be done whenever $x > -1/2$. 

For every $x: |x| < 1/2$, we now apply \fctref{fct:derivative} to functions $F(x)$ and $\frac{t}{\mu \sqrt{1+x}}$ and use \fctref{fct:sqroot} to write:

$$ \max_{\ell \leq k} \sup_{ |x| < 1/2} |g^{(\ell)}(x)| \leq  (k!)^3 \cdot \left(\frac{t}{\mu}\right)^k \cdot \left(\max_{\ell
    \leq k} \sup_{  -\frac{1}{2} \leq x \leq \frac{1}{2}} \left|F^{(\ell)}\left(\frac{t}{\mu \cdot \sqrt{(1+x)}}\right)\right|\right) \eqdef \zeta(t).$$

Let $P_k:\R \to \R$ be the degree $k-1$ polynomial obtained by
truncating the Taylor expansion of $g$ at $0$.
Thus, each coefficient of $P_k$ is at most $\zeta(t)$ for $z > -1/2$. Thus, using the error term of the Taylor expansion, we have:

\begin{equation}
  \label{eq:genproof1}
|g(z) - P_k(z)| < k\zeta(t) \cdot |z|^k. \end{equation}

%
%
%
%
We can now write the CDF distance between $\sqrt{X} \cdot Z$ and $\sqrt{Y} \cdot Z$ as:
\begin{multline}\label{eq:genproof2}
 \left|\E[F(\frac{t}{\sqrt{X}})] - \E[F(\frac{t}{\sqrt{Y}})]\right| = \left|\E[g(\hat{X})] - \E[g(\hat{Y})]\right| \leq \\
\left|\E[g(\hat{X})] - \E[P_k(\hat{X})]\right| + \left|\E[P_k(\hat{X})] - \E[P_k(\hat{Y})]\right| + \left|\E[P_k(\hat{Y})]-\E[g(\hat{Y})]\right|.
\end{multline}
We will bound each term in the right-side individually. We first use the
fact that $X$ and $Y$ are $2k$ moment matching to bound the middle
term:
\begin{equation}
 \left|\E[P_k(\hat{X})] - \E[P_k(\hat{Y})]\right| \leq \epsilon \cdot \|P\|_1 \leq k \zeta(t) \epsilon. \label{eq:middle-term}
\end{equation}

Next, we bound the first term of \eqref{eq:genproof2}. Let $\1(G_X)$ be the $0$-$1$ indicator of the event $|\hat{X}| < 1/2$ and set $\1(\neg \G_X) = 1-\1(\G_X)$. Then,
\begin{align}
|\E[g(\hat{X})] - \E[ P_k(\hat{X})] | &\leq \E\left[\left|g(\hat{X}) - P_k(\hat{X})\right|\right]\\
&= \E\left[\left|g(\hat{X}) - P_k(\hat{X})\right|\cdot 1(\G_X)\right]+
\E\left[\left|g(\hat{X}) - P_k(\hat{X})\right| \cdot 1(\neg \G_X)\right]
\nonumber \\
&\text{ When $\1(\G_X) = 1$, we can use \eref{eq:genproof1} to bound the first term (recall that $k$ is even):}  \nonumber \\
&\leq \E\left[k\zeta(t) \cdot \hat{X}^k \right] + \E\left[\left|g(\hat{X}) - P_k(\hat{X})\right| \cdot 1(\neg \G_X)\right]\\
&\text{ Applying Cauchy-Schwartz to the second term:} \nonumber \\
&\leq \E\left[k\zeta(t) \cdot \hat{X}^k \right]+
\E\left[\left|g(\hat{X}) - P_k(\hat{X})\right|^2\right]^{1/2} \cdot
\E[1(\neg G_X)]^{1/2} \nonumber \\
&\leq k\zeta(t) \cdot \frac{\mu_k}{\mu^{2k}} + 2\left(1 +
  \E[P_k(\hat{X})^2]^{1/2}\right) \cdot \E[1(\neg(G_X))]^{1/2} \label{eq:better-milder}.
\end{align}
Next, we bound the last two terms of \eqref{eq:better-milder}. Using the bound on coefficients of $P_k$ and the upper bound of $1$ on the moments of $\hat{X}$:
\begin{equation} \label{eq:polynomial-norm}
\E[P_k(\hat{X})^2]^{1/2} \leq \zeta(t) \cdot
\left(\sum_{\ell=0}^{k-1} \E[\hat{X}^{2\ell}]^{1/2}\right)  \leq k
\zeta(t). 
\end{equation}
By Markov's inequality applied to $\hat{X}^{2k}$, 
\begin{equation}
  \label{eq:genproof3}
 \E[1(\neg \G_X)] = \pr[|\hat{X}| \geq 1/2] \leq 2^{2k} \cdot \frac{\mu_{2k}}{\mu^{4k}}.  
\end{equation}
Combining the above bounds we get, 
$$ \E\left[\left|g(\hat{X}) - P_k(\hat{X})\right|\right] \leq
k\zeta(t) \frac{\mu_k}{\mu^{2k}} + (1+ k\zeta(t))2^{k+1}
\frac{\sqrt{\mu_{2k}}}{\mu^{2k}}.$$
We now argue similarly for the case of $Y$. Let $\1(\G_Y)$ be $0$-$1$ indicator of the event $|\hat{Y}| < 1/2$ and $\1(\neg \G_Y) = 1-\1(\G_Y)$.
Now, $\E[ \hat{Y}^k] \leq \E[\hat{X}^k] + 2\epsilon \leq \frac{\mu_k}{\mu^{2k}} + 2\epsilon.$ Next, as above, an application of Markov's inequality combined with the observation that $X,Y$ are moment matching (and thus $\E[(Y-\mu^2)^{2k}] \leq \E[ (X -\mu^2)^{2k}] + 2\epsilon \mu^{4k}$) yields 
$$\E[\1(\neg G_Y)]
\leq 2^{2k} \cdot \frac{\E[ (Y-\mu^2)^{2k}]}{\mu^{4k}} \leq 2^{2k} \frac{\mu_{2k}}{\mu^{4k}} + 2^{2k+1}\epsilon.$$
Thus, arguing as in the case of \eqref{eq:better-milder}, we obtain: $$\E \left[ \left|g(\hat{Y}) -P_k(\hat{Y}) \right| \right]
\leq k\zeta(t) \frac{\mu_k}{\mu^{2k}} + 2k \zeta(t)\epsilon + 2^{k}(1+ k\zeta(t)) \cdot
\frac{\sqrt{\mu_{2k}}}{\mu^{2k}} + (1+k\zeta(t))2^{k+2}\epsilon.$$
Thus, from \eqref{eq:genproof2}, we finally have:
\begin{equation}
\left|\E[F(t/\sqrt{X})] - \E[F(t/\sqrt{Y})]\right|  \leq 2k \left( \zeta(t)
\frac{\mu_{k}}{\mu^{2k}}\right) + (1+k\zeta(t))2^{k+1}\frac{\sqrt{\mu_{2k}}}{\mu^{2k}} + \epsilon \cdot 2^{3k} \zeta(t). \label{eq:final-bound}
\end{equation}
When $t \geq 2 \mu F^{-1}(1-\delta)$, it is easy to bound the CDF distance:
\begin{align*}
  \E[F(t/\sqrt{X})] &\geq  \E[F(t/\sqrt{X}) | \1_{\G_X} = 1] \cdot \pr[1(\G_X)=1]
  \\&\geq F(t/2\mu) \cdot (1 - \pr[1(\neg\G_X)=1]) 
\geq 1 - \delta -
  \pr[1(\neg\G_X)=1]\\
& \geq 1- \delta - 2^{2k} \frac{\mu_{2k}}{\mu^{4k}}.
\end{align*}
Similarly, in this case:
  $$\E[F(t/\sqrt{Y})] \geq 1 - \delta -  2^{2k} \frac{\mu_{2k}}{\mu^{4k}} - 2^{2k+1}\epsilon.$$ Thus, when $t \geq 2 \mu F^{-1}(1-\delta)$, $$\left| \E[F(t/\sqrt{X})] -\E[ F(t/\sqrt{Y})] \right| \leq \delta + 2^{2k+1} \cdot \frac{\sqrt{\mu_{2k}}}{\mu^{4k}} + 2^{2k+1}\epsilon.$$
  
On the other hand, when $t \leq
2\mu F^{-1}(1-\delta)$ (and $|x|  \leq 1/2$) yields $\frac{t}{\mu \sqrt{1+x}}\in [2\sqrt{\frac{2}{3}}, 2\sqrt{2}]$ and thus: $\zeta(t) \leq \Delta_k$. We can now use \eqref{eq:final-bound} setting $\zeta(t) \leq \Delta_k$ to obtain the lemma.
\end{proof}

\section{PRPs from approximate orthogonal designs} \label{sec:PRPs}
In this section, we show how to construct PRPs from approximate orthogonal designs and thus proving \tref{thm:sphericalprojection}, which we first restate here.

\begin{theorem}[Theorem \ref{thm:sphericalprojection} restated]
Fix any $\epsilon > 0$. Then, for any $m = \Omega( \log^2{(1/\eps)})$, there exists an $\epsilon$-PRP from $\R^m
\rightarrow \R^{\sqrt{m}}$ with a seed-length of $O( \log{
  (m/\epsilon)})$. 
\end{theorem} 

For this, we shall assume the existence of a good explicit approximate orthogonal design, a proof of which is provided in the next section:
\begin{lemma}[\cite{BHH12}; see \sref{sec:orthdesign}]\label{thm:construct-design}
There exists an efficiently samplable $\epsilon$-approximate orthogonal $t$-design with seed-length $O( t \log{(n)} + \log{(1/\epsilon)})$.
\end{lemma}

\ignore{
\begin{lemma}
Let $\epsilon > 0$ and $n$ any positive integer. Suppose there exists an efficiently samplable $\epsilon$-approximate orthogonal $t$-design in $n$ dimensions with seed-length $O(t\log{(n)} + \log{(1/\epsilon)})$. Then, for $m = \Omega \left( \log^2{(1/\epsilon)} \right)$, there exists an $\epsilon$-PRP from $\R^m \rightarrow \R^{\tm}$ for $\tm =\sqrt{m}$ with a seed-length of at most $O( \log{(m/\epsilon)})$. \label{lem:orthdesign-PRP}
\end{lemma}}

We first describe the idea behind the proof of \tref{thm:sphericalprojection}. Let $\D$ be an $\epsilon$ approximate orthogonal $2k$-design in $m$ dimensions and denote $\sqrt{m}$ by $\tm$. For every $P' \sim \D$, we define $P \in \R^{\tm \times m}$ to be the matrix obtained by taking the first $\tm$ rows of $P'$. Let the resulting distribution on matrices in $\R^{\tm \times m}$ be denoted by $\D_{\tm}$. We will show that $\D_{\tm}$ is a $\beta$-PRP where $\beta$ depends on the parameters $k$ and $\epsilon$. Suppose $Q$ is the uniformly random projection from $\R^m \to \R^{\tm}$, $w \in \R^m$ satisfies $\|w\|=1$ and $v \sim \S^{\tm-1}$. Then, to show that $P$ is an PRP, we must show that $\dcdf(\langle Pw, v\rangle, \langle Qw, v \rangle)$ is small. 

Let $Y = \|Pw\|^2$. The random variable $\langle Pw, v \rangle$ has the same distribution as $\sqrt{Y} \cdot \langle w', v \rangle $ for $w'$:$\|w'\| = 1$ and $v \sim \S^{\tm-1}$ independent of $Y$. Thus, the random variable of our interest is a mixture $\sqrt{Y} \cdot Z$ where $Z$ is distributed as $\langle w',v \rangle$. Similarly, $\langle Qw, v \rangle$ is distributed as a mixture $\sqrt{X} \cdot Z$ where $X = \|Qw\|^2$. 

Thus, if we show that a) $Z$ has an infinitely differentiable CDF with a reasonably bounded tail and derivatives b)$X$ has sufficiently slow growing moments and c) $X, Y$ are approximately moment matching, then we can apply the result from the previous section to show that the CDF distance between $\sqrt{X} \cdot Z$ and $\sqrt{Y} \cdot Z$ is small. In the following, we implement this plan and show that the $X,Y,Z$ defined above indeed satisfy the conditions required to complete the proof of \tref{thm:sphericalprojection}. We first record the required properties of $Z$:

\begin{lemma}
Let $Z$ be the random variable distributed as $\langle w, v \rangle$ for $v \sim \S^{m-1}$ and any fixed $w$:$\|w\| = 1$ as above. Let $f, F$ be the PDF and CDF of $Z$, respectively. Then,

\begin{enumerate}
\item \textbf{ Sharp Tail:} $F^{-1}(1-\delta) < 1/10 $ for $\delta = 0.995^{\sqrt{m}}$.
\item \textbf{ Bounds on the Derivatives:} For any $0 < x < 1: |f^{(q)}(x)| \leq c \cdot 10^q \cdot q!/|x|^q$, where $c = \frac{1}{\sqrt{\pi}} \cdot \frac{\Gamma(m/2)}{\Gamma(\frac{m-1}{2})}.$ 
\end{enumerate}
\label{lem:props-of-Z}
\end{lemma}

Next, we bound the moments of $X$:
\begin{lemma}[Moments of $X$] \label{lem:moment-matching-verify}
Fix any $w \in \R^m$ such that $\|w\| =1$ and define random variables $X = \|Qw\|^2$ and $Y = \|Pw\|^2$ where $Q$ is the uniformly random projection from $\R^m \to \R^{\tm}$ for $\tm = \sqrt{m}$ and $P \sim \D$. Then, for any $p \leq \tm/20$, $$\E[ \frac{ |X-\E[X]|^{p}}{\E[X]^p}] \leq p^{O(p)} \cdot {\tm}^{-4p/5}.$$
\end{lemma}

We can now complete the proof of \tref{thm:sphericalprojection} using \lref{lem:moment2cdf}, \lref{lem:props-of-Z} and \lref{lem:moment-matching-verify}.

\begin{proof}[Proof of \tref{thm:sphericalprojection}]
Let $k < \tm/80$ be such that $(\tm/k)^{-\Theta(k)} < \epsilon$. Since $m = \Omega( \log^2{(1/\epsilon)})$, such a $k$ exists. Let $\D$ be an $\epsilon$-approximate $4k$-orthogonal design in $m$ dimensions and $\D_{\tm}$ be the distribution over matrices in $\R^{\tm \times m}$ obtained by taking the submatrix of first $\tm$ rows of a random draw from $\D$. Let $Q$ be a uniformly random projection matrix obtained by taking the first $\tm$ rows of a Haar distributed matrix from $\SO(m)$.

We claim that the map: $w \in \R^m \to Pw \in \R^{\tm}$ where $P \sim \D_{\tm}$ is an $\epsilon$-PRP. We set $X = \|Qw\|^2$, $Y = \|Pw\|^2$ and $Z$ distributed as $\langle w, v \rangle$ for $v \sim \S^{\tm-1}$. Observe that by our construction of $P$ and $Q$ above, each of $X,Y$ can be seen as polynomial of degree $2$ on $\SO(m)$. Thus, $X^{t}$ and $Y^{t}$ are are polynomials on $\SO(m)$ of degree at most $2k$ for any $t \leq k$. It is easy to observe (using that $\sum_{i} w_i^2 = 1$) that the absolute value of the coefficients of the coefficients of these polynomials is at most $m^{2k}$. Further, since $\D$ is an $\epsilon$-approximate $4k$-design, we must have, for any univariate polynomial $p$ of at most $k$ and sum of absolute value of its coefficients $||p||_1 = 1$: $$\left| \E[ p(X)] - \E[ p(Y)] \right| \leq \frac{\epsilon}{m^{4k}} \cdot m^{2k} \leq \epsilon \cdot m^{-2k}.$$

We use \lref{lem:moment-matching-verify} and \lref{lem:props-of-Z} to obtain the estimates of all parameters required to apply \lref{lem:moment2cdf} to $X$ and $Y$ above now. Let $F$ and $f$ be the CDF and PDF of $Z$ respectively. For $X$, $Y$ and $Z$ above, we have, $\Delta_k = k^{\Theta(k)}$ and $\mu_k/\mu^{2k}$, $\sqrt{\mu_{2k}}/\mu^{2k} \leq m^{-\Theta(k)}$. We set $\delta = 0.995^{\sqrt{m}}$ and note that $F^{-1}(1-\delta) < 1/10$. It is easy to verify that with this setting of the parameters, \lref{lem:moment2cdf} gives an error bound of $\epsilon$ for appropriate setting of constants hidden in the $\Theta$s.

Next, we verify the seed-length used for the construction above: observe that the $k$ chosen above can be written as $\max \{ \Theta(1), \Theta\left(\log{(1/\epsilon)}/\log{(m)}\right) \} $. Thus, the required seed-length is given by: $O(k \log{(m)} + \log{(1/\epsilon)}) = O( \log{(m/\epsilon)})$ as promised.
\end{proof}
In the remaining part of this section, we prove \lref{lem:props-of-Z} and \lref{lem:moment-matching-verify}.
\begin{proof}[Proof of \lref{lem:props-of-Z}] 
Because of rotation invariance of $v \sim \S^{\tm-1}$, $Z$ is identically distributed as the first coordinate of $v$: $v_1$. The PDF $f:[-1,1] \rightarrow [0,1]$ of $Z$ at any $t$ is given by (we give a proof in the \sref{app:sec-PRP} of the Appendix): $$f(t) = c_{\tm} \cdot
(1-t^2)^{\frac{\tm-3}{2}}$$ where $c_{\tm}  = \frac{1}{\sqrt{\pi}} \cdot
\frac{\Gamma(\tm/2)}{\Gamma(\frac{\tm-1}{2})}.$ From the expression for $f$, it is easy to verify that for $\delta = 0.995^{\sqrt{m}}$, $F^{-1} (1-\delta) < 1/10$. 

To bound the derivatives of $f$, we will need the following standard theorem from complex analysis (due to Cauchy): 
\begin{fact}[Cauchy's Estimate]
Let $\eta:\C \rightarrow \C$ be a holomorphic function on $D(z,r)$, the closed disk of radius $r$
centered at $z \in \C$, and let $M = \max
_{y: |z-y|=r} |F(y)|$. Then,  $$|\eta^{(k)}(z)| \leq \frac{k!}{r^k} \cdot M.$$ 
\end{fact}

We can now estimate the derivatives of $f$. We first extend $f$ to  all of the
complex plane to obtain $\tf$. Then, $\tf(z) = c_{\tm} \cdot
(1-z^2)^{\frac{\tm-3}{2}}$ is a polynomial in $z$ and thus holomorphic
everywhere and in particular, on every closed disk in the complex
plane. The complex derivatives of $\tf$ at any point $x$ on the real line
are equal to the derivatives of $f$ at $x$. Thus, we can apply Cauchy's estimate to produce an upper bound on the derivatives.

Consider the disk $D(x,r)$ of radius $r=x/10$
centered at a (real) $x$ in the complex plane. We need to estimate the maximum absolute value of $\tf$ over
the boundary of $D(x,r)$. Parametrize the boundary of
$D(x,r)$ as $z = x+ rcos(\theta) + i r sin(\theta)$ for $\theta \in [0,2\pi]$. We can
write $|1-z^2| = 1- 2Re(z^2) + |z|^4$. Now, $Re(z^2) = (x + r
cos(\theta))^2 - r^2  sin^2(\theta) \geq 0.8 x^2$. On
the other hand, $|z|^4 = (x+ r cos(\theta))^2 + r^2
sin^2(\theta))^2 \leq (x+r)^4 \leq 1.21x^4$. Thus, $-2Re(z^2) +
\|z\|^4 \leq -1.6x^2+1.21 x^4 \leq 0$. Thus, $|1-z^2| \leq 1$ for every $z$ on the boundary of
$D(x,r)$. Hence, the maximum absolute value of $\tf$ on $D(x,r)$ is at
most $c_{\tm}$. Applying Cauchy's estimate to $\tf$ now gives: $$|f^{(q)}(x)| \leq c_{\tm} \cdot 10^{q} \cdot q!/|x|^{q}.$$

\end{proof}

Before moving on to prove \lref{lem:moment-matching-verify}, we collect three facts useful in the proof:

We will need the following Marcinkiewicz-Zygmund inequality for moments and the standard gaussian concentration:
\begin{fact}[Marcinkiewicz-Zygmund \cite{MZ64} (see also \cite{wiki:MZ})]
Let $S_1,S_2, \ldots, S_q$ be a sequence of independent, zero mean random variables. Then: 
$$\E[ |\sum_{i = 1}^q S_i|^p] \leq p^{O(p)} \cdot q^{p/2-1} \sum_{i = 1}^q \E[ |S_i|^p].$$
\label{fact:moment-inequality}
\end{fact}
The following concentration bound is standard for Gaussian random vectors.
\begin{fact}[Gaussian Concentration \cite{ledoux2005concentration} (see also Corollary 2.3 on Page 6 of \cite{Barvinok05}) ]
Let $g\in \R^m$ be a vector with each coordinate distributed independently as $\cN(0,1)$. Then, for any $0 < \xi < 1$: $$\Pr[ \|g\|^2-m| > \xi \cdot m] \leq 2e^{-\xi^2 \cdot m/4}.$$

\end{fact}

\begin{proof}[Proof of \lref{lem:moment-matching-verify}] 
Due to the rotation invariance of the uniformly random projection $Q$, $X = \|Qw\|^2$ has the same distribution as the Euclidean norm of the vector formed by the first $\tm = \sqrt{m}$ coordinates of $v \sim \S^{m-1}$. Thus, $\E[X]  = \tm/m = 1/\sqrt{m}$. Recall that $v$ has the same distribution as the random variable $g/\|g\|$ where $g\in \R^m$ has each coordinate distributed $\cN(0,1)$. Let $\1_\G$ be the indicator of the event that $m(1-\xi) \leq \|g\|^2 \leq (1+\xi) m$ and $\neg \1_\G$, its negation. Then, we have: 
\begin{align*}
\E[ |X-\E[X]|^{p}] &= \E[ |\sum_{i = 1}^{\tm}v_i^2-\tm/m|^{p}] = \E[ \frac{|\sum_{i = 1}^{\tm}g_i^2 - \|g\|^2 \cdot \frac{\tm}{m}|^{p}}{\|g\|^{2p}}]  \\
&=   \E[ \frac{|\sum_{i = 1}^{\tm}g_i^2 - \|g\|^2 \cdot \frac{\tm}{m}|^{p}}{\|g\|^{2p}} \cdot \1_\G] +  \E[ \frac{|\sum_{i = 1}^{\tm}g_i^2 - \|g\|^2 \cdot \frac{\tm}{m}|^{2p}}{\|g\|^{2p}} \cdot \neg \1_\G].\\
\text{ Estimating first } & \text{term  using the bound on $\|g\|^2$ } \\
\text{ and  applying} & \text{ Cauchy-Schwartz to the second term,}\\
&\leq  (\frac{1}{1-\xi})^p \cdot \E \left[ \frac{2^p \left(\left|\sum_{i = 1}^{\tm}g_i^2 - \tm \right|^{p} + \xi^p \cdot \tm^p \right)}{m^p} \right] + \E[ \neg \1_\G]^{1/2} \cdot \left(\E\left[ \left|\sum_{i = 1}^{\tm}v_i^2-\tm/m \right|^{2p}\right] \right)^{1/2}.\\
\text{ Using that } & \text{ $\E[v_i^2] = 1/m$ for each $i$ }\\
&\leq  (\frac{2}{m(1-\xi)})^p \cdot (\xi^p {\tm}^p + \E[|\sum_{i = 1}^{\tm} (g_i^2-1)|^{2p}]) + 2^{p} e^{-\xi^2m/8} \cdot (\frac{\tm}{m})^{p}.\\
\end{align*}

The first term is easy to estimate using \fctref{fact:moment-inequality} (the same result can also be obtained from Bernstein-like inequalities for exponential random variables \cite{Ver12}). We have:
$$ \E[|\sum_{i = 1}^{\tm} (g_i^2-1)|^{p}] \leq p^{O(p)} \cdot {\tm}^{p/2} \cdot \E[ |g_i^2-1|^{p}] = p^{O(p)} \cdot {\tm}^{p/2}.$$
Choosing $\xi = 1/m^{1/5}$ and noting that $p < \tm/20$, we can now obtain: $\E[ |X - \E[X]|^p]/(\E[X])^p \leq p^{O(p)} \cdot m^{-2p/5} = p^{O(p)} \cdot {\tm}^{-4p/5}.$

\end{proof}
\section{Constructing approximte orthogonal designs} \label{sec:orthdesign}
In this section, we give a proof of a construction of approximate orthogonal designs, proving \lref{thm:construct-design}, based on a recent result of Bourgain and Gamburd \cite{BG11}. This also follows from the work Brandao et.~al.~\cite{BHH12} (Page 17, Equation B2) except for some technicalities  and concrete quantitative bounds which we need and work out next.  
\ignore{
The following is the main result of this section.

\begin{lemma} \label{thm:construct-design}
There exists an efficiently samplable $\epsilon$-approximate orthogonal $t$-design with seed-length $O( t \log{(n)} + \log{(1/\epsilon)})$.
\end{lemma}}
We first provide some background before stating the result of \cite{BG11} (see the text by Bump \cite{Bu11} for a detailed exposition).

Bourgain and Gamburd\cite{BG11} show that there exist Cayley graph expanders on $\SU(n)$. This
also implies that there exist Cayley graph expanders with the same
parameters on the group $\SO(n)$ (see \aref{app:liegroups}). In this section, we use the
construction for $\SO(n)$ to obtain approximate orthogonal $t$-designs. We
provide the necessary background and the deferred proofs from this section in \aref{app:liegroups}. 

We briefly recall Cayley graphs on finite groups before working on $\SO(n)$. A Cayley graph on a group $G$ is defined by a set of generators (inverse closed) $g_1, g_2, \ldots g_k$. The vertex set is given by the elements of the group $G$ and there is an edge between $h,h'$ iff $h= g_i h'$ for some generator $g_i$. 

One can define a Cayley graph on an infinite group similarly and in the following, we adopt the linear operator view.

Let $\L^2(\SO(n)) = \{ \rho:\SO(n) \rightarrow \R \mid \int_{\H} \rho^2 < \infty \}$ be the linear space of square integrable functions on $\SO(n)$ with respect to the Haar measure $\H$. For any $g \in \SO(n)$, let $\rT_g: \L^2(\SO(n)) \rightarrow \L^2(\SO(n))$ be the linear operator on $\SO(n)$ defined so that $\rho \rightarrow \rT_g \cdot \rho$, where $(\rT_g \cdot \rho)(y) = \rho( y \cdot g)$. 

Next, we define Hecke (or averaging) operators on $\L^2(\SO(n))$ that correspond to the finite dimensional normalized adjacency matrices described above. We say a set $Gen = \{g_1, g_2, \ldots, g_k\}$ is inverse closed if for every $g \in Gen$, $g^{-1} \in G$.
\begin{definition}[Hecke Operators]
For some universal constant $k>0$, an averaging operator (also known as \emph{Hecke} operator) with an inverse closed set of generators $g_1, g_2, \ldots, g_k \in G$ is a linear operator $\rT: \L^2(G) \rightarrow \L^2(G)$ defined by $\rT \eqdef \frac{1}{k} \sum_{i = 1}^k \rT_{g_i}$.
\end{definition}
It is easy to verify that a Hecke operator $\rT$ on $\L^2(\SO(n))$ is bounded and compact and thus has a spectrum. Thus we can look at the gap between the first and second eigenvalues of $\rT$ to talk of the expansion of the associated graph. This is encapsulated in the following definition:

\begin{definition}[Hecke Operators with Spectral Gap]
A Hecke operator on $\L^2(G)$ for a compact Lie group $G$ is said to have a spectral gap, if for every $\rho \in \L^2(G)$ such that $\|\rho\|_2 = 1$ and $\E_{g \sim \H}[ \rho(g)] = \int_{\H} \rho dg = 0$, $\|\rT \rho\|_2  \leq \lambda < 1-\delta$ for a universal constant $\delta > 0$.
\end{definition}

We can now describe the (consequence of) result of Bourgain-Gamburd \cite{BG11} that we need in the language of Hecke operators:
\begin{corollary}[See \cref{corBG}]
For a universal constant $k > 0$, there is an explicit Hecke operator $\rT$  with $k$ generators on $\L^2(SO(n))$ with a spectral gap $1-\lambda$ bounded away from $0$. 
\end{corollary}

We now show how to obtain orthogonal $t$-designs using the above
corollary. The idea itself is standard (see for example \cite{LPS87}) and we again use the intuition for finite graphs to motivate it: imagine running a random walk on the Cayley graph for a few ($\sim \log{(1/\epsilon)}$) steps. In the finite dimensional world, we expect that the resulting distribution on the vertices of the graph to be ``close" to ($\sim \epsilon$) uniform. 

In our setting, recall that our aim is to construct an object that fools the Haar (``uniform") distribution on $\SO(n)$. If we start a ``random walk" on a Cayley graph with generators $g_1, g_2, \ldots, g_k$ on $\SO(n)$ from some fixed point, we expect the resulting distribution to be close to ``uniform" on $\SO(n)$ after a few steps. \ignore{ (in a strict technical sense, we'd need that the graph be dense inside $\SO(n)$, but the existence of such expanders is guaranteed by \cite{BG11}).}
This argument can be formalized to yield $\epsilon$-approximate orthogonal $1$-designs, i.e. those that fool all linear functions in the entries of the matrices drawn according to the Haar distribution on $\SO(n)$. To fool higher degree polynomials, we first take the tensor powers of the generators of the ``Cayley graph" on $\SO(n)$. The entries of $g_i^{\otimes t}$, the $t$-wise tensor (Kronecker) product of $g_i$ with itself, are all monomials of degree at most $t$ in the entries of $g_i$. Thus if we start with the Cayley graph with the generators given by the $t^{th}$ tensor powers of $g_i$, and argue similarly as above, we should hope to get approximate orthogonal $t$-designs.

We now show how to formalize this argument starting from the Hecke operator $\rT$ given by Bourgain and Gamburd $\cite{BG11}$. Suppose $g_1, g_2, \ldots, g_k$ are generators of the Hecke operator $\rT$ as in the above corollary. Then, the constant function $\1:\SO(n) \rightarrow \R$, $\1(g) = 1$ for every $g \in \SO(n)$ is an eigenfunction of $T$ with eigenvalue $1$. In particular, the operator $\rT'$ defined by $$\rT' \rho \eqdef \rT \rho - \E_{g \sim \H}[ \rho(g)],$$ has a spectrum and all its eigenvalues are at most $\lambda$.  

Define $$S \eqdef 1/k \cdot \sum_{i = 1}^k g_i^{\otimes t} - \E_{g \sim \H}[ g^{\otimes t}].$$ Then, $S \in \R^{n^t \times n^t}$. We show that the spectral norm of $S$, $\|S\| \leq \lambda$.

\begin{lemma} \label{lem:tensor-product-expanders}
Let $\rT$ be a Hecke operator on $\L^2(\SO(n))$ with generators $g_1, g_2, \ldots, g_k$ with second largest eigenvalue $\lambda$. For any positive integer $t \in \N$, let $S = \left( \frac{1}{2k} \cdot \sum_{i = 1}^k g_i^{\otimes t} -\E_{g \sim \H}[ g^{\otimes t}] \right) \in \R^{n^t \times n^t}$. Then, $$\|S\| \leq \lambda.$$
\end{lemma}
\begin{proof}
Let $v \in \R^{n^t}$, $||v|| = 1$, be an eigenvector of $S$ with eigenvalue $\theta$. We will show that there is an eigenfunction $\gamma \in \L^2(\SO(n))$, $\gamma:\SO(n) \rightarrow \R$ of $\rT'$ with the corresponding eigenvalue $\theta$. Since $\|\rT'\| \leq \lambda$, we will have the result of the lemma.

Define $\gamma(g) = \ip{g^{\otimes t} \cdot v}{e_1}$ where $e_1 \in \R^{n^t}$ with the first coordinate $1$ and $0$ otherwise. It is easy to observe that $\gamma \in \L^2(\SO(n))$. We now verify that $\gamma$ is an eigenfunction of $T'$.

 For any $h \in \SO(n)$,
\begin{align*}
\rT' \cdot \gamma (h)  &= \frac{1}{2k} \cdot \sum_{i = 1}^k (\rT_{g_i} \cdot \gamma)(h) - \E_{g \sim \H} [\gamma( g)]\\
& = \frac{1}{2k}  \cdot \sum_{i = 1}^k \gamma( h \cdot g_i)  - \E_{g \sim \H} [\gamma( h \cdot g)]\\ 
& = \frac{1}{2k} \cdot \sum_{i = 1}^k \ip{h^{\otimes t} \cdot g_i^{\otimes t} \cdot v}{e_1}  - \ip{h^{\otimes t} \cdot \E_{g \sim \H} [g^{\otimes t}] \cdot v}{e_1}\\
& = \ip{h^{\otimes t} \cdot \frac{1}{2k}  \cdot \sum_{i = 1}^k g_i^{\otimes t}  \cdot v}{e_1} -  \ip{h^{\otimes t} \cdot \E_{g \sim \H} [g^{\otimes t}] \cdot v}{e_1}\\
& = \ip{h^{\otimes t} \cdot S \cdot v}{e_1} = \theta \cdot \ip{h^{\otimes t} \cdot v}{e_1} = \theta \cdot \gamma(h),
\end{align*}

\end{proof}
where we use $(h\cdot g)^{\otimes t} = h^{\otimes t} \cdot g^{\otimes t}$ (which can be proven using induction and the mixed product property of the Kronecker product).

We can now use the result above to derive the main theorem of this section.
\begin{proof}[Proof of \lref{thm:construct-design}]

Let $\rT$ be a Hecke operator on $\L^2(\SO(n))$ with the second largest eigenvalue at most $\lambda$. Then, as above, $\rT':\L^2(\SO(n)) \rightarrow \L^2(\SO(n))$ defined by $\rT' f(g) = \rT f(g) - \E_{g \sim \H}[ f]$ satisfies $\|\rT'\| \leq \lambda$. Further, for any positive integer $q$, $\|{\rT'}^q\| \leq \lambda^q$. We will choose $q$ appropriately later.

Let $g_1, g_2, \ldots, g_k$ be the generators ${\rT}^q$ and let $\D$ be a uniform draw from $\{g_1, g_2, \ldots, g_k\}$. We will show that $\D$ is an $\epsilon$-approximate orthogonal $t$-design.

From \lref{lem:tensor-product-expanders} applied to ${\rT}^q$,  we know that $S = 1/k \cdot \sum_{i = 1}^k g_i^{\otimes t} -\E_{g \sim \H} [g^{\otimes t}]$ satisfies $\|S\|\leq \lambda^q$. 
  
We must show that for every polynomial $p:\SO(n) \rightarrow \R$ of degree at most $t$ with $\|p\|_1 = 1$, $$\left|\E_{g\sim \H}[ p(g)] - \E_{g \sim \D} [p(g)]\right| \leq \frac{\epsilon}{n^t},$$ for some $q =  t \log{(n)} + \log{(1/\epsilon)}$. Observe that it is enough to show this statement for monomials.

Let $M: \SO(n) \rightarrow \R$ be any monomial of degree $d \leq t$ (i.e. $M(g)$ is a product of at most $t$ entries of the matrix $g$ for any $g \in \SO(n)$). Then, $M(g)$ is an entry of the matrix $g^{\otimes d}$. Thus, 
\begin{align*}
 |\E_{g \sim \D}[ M(g)] -\E_{g \sim \H}[ M(g)]| \leq \left\| \frac{1}{k} \sum_{i = 1}^k {g_i}^{\otimes d} - \E_{g \sim \H}[g^{\otimes d}]\right\| \leq \lambda^q.
\end{align*}
  
Thus, choosing $q =\Theta(t \log{(n)} + \log{(1/\epsilon)})$ is enough. Thus, $\D$ is an $\epsilon$-approximate orthogonal $t$-design. 

 \end{proof}
 
\section*{Acknowledgment}
We thank the anonymous reviewers for their suggestions on better presentation of the paper and pointing out the typos in a previous version.

\bibliography{refs}
\appendix

\section{Deferred proofs} \label{app:defproofs}

\subsection{Proofs of facts from \sref{sec:moments}}
\begin{fact}[ \fctref{fct:derivative} restated]
 Let $g,h:\R \to \R$ be infinitely differentiable functions. Then,
$$|(g \circ h)^{(k)}(x)| \leq (k!)^2 \cdot \left(\max_{\ell \leq k} |g^{(\ell)}(h(x))|\right) \cdot \left(\max_{\ell \leq k} |h^{(\ell)}(x)|\right)^k.$$
\end{fact}
\begin{proof}
The result follows from using the following Faa di Bruno's formula (\cite{wikiFaaDiBruno}) for derivatives of composition of functions (whenever all the derivatives in the expression exist):

$$ \frac{d^m}{dt^m} g(h(t)) = \sum \frac{m!}{b_1! b_2! \ldots b_m!} g^{(\sum_{i = 1}^m b_i)} (h(t)) \Pi_{i=1}^m \left(\frac{h^{(i)}}{i!}\right)^{b_i},$$
where the sum is over non-negative integers $b_1, b_2, \ldots, b_m$ such that $\sum_{i = 1}^m ib_i = m$.

\end{proof}

\begin{fact}[ \fctref{fct:sqroot} restated]
For $h:(-1,\infty) \to \R$ be defined by $h(x) = 1/\sqrt{1+x}$. Then, $|h^{(k)}(x)| \leq k!$ for $-1/2 < x$. 
\end{fact}
\begin{proof}
$$h^{(k)}(x) = \frac{-1}{2}\cdot \frac{-3}{2} \cdots \frac{-(2k-1)}{2} \left( \frac{1}{1+x} \right)^{-(2k+1)/2}.$$
One can upper bound the expression on the RHS in absolute value by $k!$ for $x > -1/2$.

\end{proof}

\begin{fact} [ \fctref{lem:basic} restated]
Let $F:\R \rightarrow [0,1]$ be the CDF of a random variable $Z$. For a non-negative random variable $V$
independent of $Z$, the CDF $G$ of $V \cdot Z$ at any $t \in \R$ is given
by:$$G(t) =
\E_{X}[ F(t/V)].$$
\end{fact}
\begin{proof}
Let $\eta$ be the PDF of $V$.  We have: $$\Pr[V \cdot Z \leq t] =
\int_{0}^{\infty} Pr[ Z \leq t/x] \cdot \eta(x) dx =
\int_{0}^{\infty} F[t/x] \eta(x) dx = \E_{V}[F(t/V)].$$ 

\end{proof}

\begin{lemma}[ \lref{lm:productrv} restated]
  Let $U, V_1, V_2$ be independent random variables. Then, $\dcdf(U \cdot V_1, U \cdot V_2) \leq \dcdf(V_1, V_2)$. 
\end{lemma}

\begin{proof}

We can write the CDF of $U \cdot V_1$ as: $$F_{UV_1} (t) = \Pr[ UV_1
\leq t]  = \int_{0}^{\infty} F_{V_1}(t/u) f_U(u) du + \int_{-\infty}^0
(1-F_{V_1}(t/u) f_U(u) du.$$

Similarly, the CDF of $U \cdot V_2$ can be written as : $$F_{UV_2}(t)
= \int_{0}^{\infty} F_{V_2}(t/u) f_U(u) du + \int_{-\infty}^0
(1-F_{V_2}(t/u) f_U(u) du.$$

Thus, $\dcdf(U \cdot V_1, U \cdot V_2) \leq \int_{0}^{\infty}
\dcdf(V_1, V_2) f_U(u) du + \int_{-\infty}^{0} \dcdf(V_1, V_2) f_U(u)
du = \dcdf(V_1,V_2).$

\end{proof}

\subsection{Proofs from \sref{sec:PRPs}} \label{app:sec-PRP}
Recall that for the Gamma function $\Gamma$, $\Gamma(1/2) = \sqrt{\pi}$. The surface area of unit sphere $\S^{n-1}$ is $\frac{2\pi^{n/2}}{\Gamma(n/2)}$. We now obtain the PDF of the distribution of $<w,v>$ for $v \sim \S^{n-1}$. We now compute the PDF of $\langle w, v \rangle $ for $v \sim \S^{n-1}$.

\begin{fact}
Let  $x = \ip{w}{v}$ for $v$ distributed uniformly over $\S^{n-1}$ and $w \in \R^n$, fixed, with $\|w\|= 1$. Then, $x$ is supported on $[-1,1]$ with the PDF: $$f(x) = \frac{1}{\sqrt{\pi}} \cdot \frac{\Gamma(n/2)}{\Gamma(\frac{n-1}{2})} \cdot (1-x^2)^{\frac{n-3}{2}}.$$
\end{fact}

\begin{proof}

Since the distribution of $x$ is invariant under any rotation of both the vectors, we can assume that $w$ has $1$ in its first coordinate and $0$ otherwise. Thus, $\ip{w}{v} = v_{1}$. We can now calculate the CDF $F$ of $x$ by $F(x) = \Pr[ v_1 \leq x]$. 

Let $L_x \subseteq \S^{n-1}$ be defined by $L_x = \{z \in \S^{n-1} \mid z_1 \leq x\}$. Let $|L_x|$ denote the surface area of $L_x$. Then, $F(x) = |L_x|/|\S^{n-1}|$. 

Let $t \in [-1,x]$. Then, $\{ z \in L_x \mid z_1 = t\}$ defines a sphere of radius $\sqrt{1-t^2}$ in $n-2$ dimensions. Using the Jacobian of the area measure $1/\sqrt{1-t^2}$, we can write $|L_x|$ as the integral:

$$|L_x| = \int_{-1}^{x}  \frac{2 \pi^{\frac{n-1}{2}}}{\Gamma(\frac{n-1}{2})} \cdot (1-t^2)^{\frac{n-3}{2}} dt.$$

Thus, $$F(x) = |L_x|/|\S^{n-1}| = 1/\sqrt{\pi} \cdot \int_{-1}^x \frac{\Gamma(\frac{n}{2})}{\Gamma(\frac{n-1}{2})} (1-t^2)^{\frac{n-3}{2}} dt.$$

Now, $$f(x) = F'(x) = \frac{\Gamma(\frac{n}{2})}{\Gamma(\frac{n-1}{2})} (1-x^2)^{\frac{n-3}{2}}.$$

\end{proof}

%
%
%
%




\section{Hecke operators with spectral gap on $\SO(n)$} \label{app:liegroups}
In this section, we show that there exist Hecke operators on the group $\L^2(\SO(n))$ with a uniform spectral gap. This result follows almost immediately from the work of Bourgain and Gamburd \cite{BG11}, who show the existence of such operators on $\L^2(\SU(n))$. For completeness, we give a straightforward argument that uses only a few standard facts from the theory of Lie groups. 

We state a few standard preliminary results (without proof) below. This material can be found in any standard textbook on Lie groups such as Bump \cite{Bu11}.

\subsection{Preliminaries}

A \emph{topological group} $G$ is a group with an underlying topology on it such that the group operation $\cdot:G \times G \rightarrow G$  is a continuous map (with respect to the induced product topology on $G \times G$). \emph{Lie groups} are topological groups where the group operation is \emph{smooth}, that is, it has derivatives of all orders. Well known matrix groups, such as the General Linear group $GL_n(\R)$: i.e. the group of invertible $n\times n$ matrices on $\R$ and its subgroups, $SL_n(\R)$: the subgroup of matrices with determinant $1$ and $\SO(n)$: the subgroup of orthogonal matrices (referred to as the \emph{rotation} group of the $n$-sphere) are all Lie groups. Similarly, the corresponding groups on the complex field $\C$: $GL_n(\C)$, $SL_n(\C)$ and $\SU(n)$ are also Lie groups. 

By viewing $\C$ as a two dimensional vector space over $\R$, we observe that $\SO(n)$ is a subgroup of $\SU(n)$. On the other hand, by observing that $\C$-linear maps on $\C^n \simeq \R^{2n}$ are strict subsets of $\R$-linear maps on $\R^{2n}$ (under the vector space transformation from $\C$ to $\R$), we observe that $\SU(n)$ is a subgroup of $SO(2n)$. 

Under the standard Euclidean topology, $\SU(n)$ and $\SO(n)$ are compact subgroups of $GL_{2n}(\R)$. They are also closed, as is evident from the fact that $\SU(n)$ and $\SO(n)$ both can be defined as subgroups of matrices with certain polynomial equality constraints on their entries and are thus inverse images of closed sets under a continuous map.

\subsubsection{Haar measure and linear operators on $\L^2(G)$} 
\label{app:sec-Haar}
For compact Lie groups $G$ such as $\SU(n)$ and $\SO(n)$, there exists a probability measure, known as the \emph{Haar} measure, $\H$ (we also use $dg$ to denote infinitesimal on $G$ with respect to $\H$), on $G$ that is invariant under (right or left) multiplication by any group element (this is written as the property of being (right or left) $G$-invariant).  When both the right and left $G$-invariant probability measures coincide, $G$ is said to be \emph{unimodular}. It is a well known fact that both $\SU(n)$ and $\SO(n)$ are unimodular groups.

We can use the Haar measure to define integrals and (Hermitian) inner products of any two functions $p,q:G \rightarrow \C$: $$ \ip{p}{q} = \int_{\H} p \cdot \bar{q} dg = \E_{\H} [ p \cdot \bar{q}], $$ where $\bar{q}$ denotes the complex conjugate of $q$. Similarly, we can define the $\ell_2$ norm of a function $f \in \L^2(G)$ by setting $\|\rho\| =\sqrt{ \int_{\H} f \cdot \bar{f} dg }= \ip{f}{f}$.

We can now define the linear space of square integrable functions on $G$: $\L^2(G) = \{ f: G \rightarrow \C \mid \E_{g \sim \H} [f(g) \cdot \bar{f}(g)] < \infty \}$. This is a Hilbert space under the inner product defined above and forms, what is known as the \emph{regular representation} of the group $G$ under the (left or right) shift action. That is, for every $g \in G$, one can define a linear operator on $\L^2(G)$, $T_g$ such that $(T_gf)(h) = f(h \cdot g)$ for every $h \in G$. Further, $G$ as a group is \emph{homomorphic} to the group of all such linear maps $\{T_g \mid g \in G\}$ under composition.

%
%
%

\subsubsection{Push forward Haar measure on coset space of closed subgroups}
Let $G$ be a unimodular Lie group with a Haar measure $\H_G$ and let $H$ be a  closed unimodular subgroup of $G$ with the Haar measure $\H_H$ defined on it. One can relate integrals of functions on either groups by constructing a $G$-invariant (probability) measure $\H_{G/H}$ on the coset space $G/H$. The existence of such a measure for compact unimodular Lie groups (a \emph{push forward} measure) is a non-trivial but well known fact. 

Let $f \in \L^2(G)$. Then, we have: $$\int_{\H_G} f dg = \int_{\H_{G/H}} \int_{\H_{H}} f( h \cdot \dot{g}) dh d \dot{g},$$ where we use $\dot{g}$ to refer to the canonical element of the coset from $G/H$ with $\dot{g}$ in it. 

This, in particular lets us define a measure preserving embedding of $\L^2(H)$ into $\L^2(G)$. Let $\rho \in \L^2(H)$ with $\|\rho\|_2 = 1$. Choose canonical elements from $G$ for every coset of $H$ in $G$. If $g$ belongs to the coset represented by the element $\dot{g} \in G$, then, $g = h \cdot \dot{g}$ for a unique element $h \in H$. We define $\tilde{\rho}: G\rightarrow \C$ by $\tilde{\rho}(g) =  \rho(h)$ in this case. Observe that the construction is well defined as any two cosets of a subgroup inside a group are either disjoint or equal. We have thus created a map $f \rightarrow \tf$. It is easy to verify that the map is measure preserving using the relationship between the integrals on $H$ and $G$ given above:

$$\int_{\H_G} \trho dg = \int_{\H_{G/H}} \int_{\H_H} \trho(h \cdot \dot{g})  dh d \dot{g} = \int_{\H_{G/H}} \int_{\H_H} f( h)  dh d \dot{g} = \int_{\H_H} \rho(h) dh.$$

In particular, $\|\trho\| = \|\rho\|$ and thus, the map defined above takes every $\rho \in \L^2(H)$ into a $\trho \in \L^2(G)$. 

\subsection{ Existence of Hecke operators with spectral gap on $\SO(n)$}

We are now ready to describe the existence of Hecke operators with spectral gap on $\SO(n)$. Bourgain and Gamburd \cite{BG11} show the following:
\begin{theorem}[Bourgain-Gamburd]
For a universal constant $k > 0$, there is a Hecke operator with a spectral gap $\rT$ on $\L^2(\SU(n))$ with $k$ generators.
\end{theorem}

Using the standard machinery developed in the preliminaries above the same result can be shown to hold for $\SO(n)$:

\begin{corollary}[Hecke Operators with Spectral Gap on $\SO(n)$]
For a universal constant $k > 0$, there is a Hecke operator with a spectral gap $\rT$ on $\L^2(\SO(n))$ with $k$ generators.\label{corBG}
\end{corollary}

\begin{proof}
Recall that we see both $G = \SU(n)$ and $H = \SO(n)$ as \emph{real} Lie groups by thinking of $\C$ as a two dimensional vector space over $\R$. Let $\rT$ be the Hecke operator on $\L^2(G)$ with an inverse closed set of generators $g_1, g_2, \ldots, g_k \in SU(n)$ with a spectral gap given by the theorem above. That is, $\rT = \frac{1}{k} \sum_{i = 1}^k \rT_{g_i} $. 

The idea of the proof is the following: We want to define an operator $\rT_{h}$ corresponding to $T_{g}$ for any generator (or its inverse) $g$. We then want to argue that $\tilde{\rT}:\L^2(H) \rightarrow \L^2(H)$ defined by $\tilde{\rT}  = \frac{1}{k} \sum_{i = 1}^k \rT_{h_i}$ where $h_i$ correspond to $g_i$ also has a spectral gap. To do this, we will use the embedding $\mathcal{E}$ of $\L^2(H)$ into $\L^2(G)$ defined above and if $\rho \rightarrow \trho$ is the embedding, then, we will show that $\tilde{\rT}\trho$ corresponds to $\rT' \rho$ in it. The spectral gap property for $\tilde{\rT}$ will then follow from the spectral gap of $\rT$.

We first define $\tilde{\rT}$. Fix the canonical elements for every coset of $H$ in $G$. Thus, $g = h \cdot \dot{g}$ whenever $g$ belongs to the coset represented by $\dot{g}$. Let $h_g \in H$ be defined by $h_g = g \cdot \dot{g}^{-1}$. Set $\tilde{\rT} = \frac{1}{k}\rT_{h_{g_i}} $.

As before, let $\rho \in \L^2(H)$ such that $||\rho||_2 = 1$ and $\int_{\H_H} \rho(h) dh = 0$.  Define $\trho \in \L^2(G)$ by $\trho(g) = \rho( g \cdot \dot{g}^{-1})$. Then from the discussion above, we know that $||\trho||_2 = 1$ and $\int_{\H_G} \trho(g) dg = 0$. Further, observe that $\rT_g \trho = \mathcal{E}(\rT_{h_g} \rho)$. 

Define $Gen_G = \{ g_1, g_2, \ldots, g_k\}$ as the (inverse closed) set of generators of $\rT$ and let $Gen_H = \{ h_{g} \mid g \in Gen_G\}$ be the corresponding set of generators for $\tilde{\rT}$ as defined above. 

We have: 

\begin{align*}
\lambda^2 \geq \int_{\H_G} (\rT\tilde{\rho})^2(g) dg &= \frac{1}{2k} \sum_{g \in Gen_G} \int_{\H_{G/H}} \int_{\H_H} (\rT_{g}\trho)^2( h \cdot \dot{g}) dh d\dot{g}\\
& = \frac{1}{2k} \sum_{h_g \in Gen_H} \int_{\H_H} (\rT_{h_{g}}\rho)^2 (h) dh = ||\tilde{\rT}\rho||_2^2.\end{align*}

This completes the proof that $\tilde{\rT}$ is a Hecke operator on $\L^2(H)$ with a spectral gap.

\end{proof}

\end{document}